\documentclass{llncs}
\usepackage{amssymb,amsmath}
\usepackage{graphicx}

\usepackage{multirow}

\title{Algorithms parameterized by vertex cover and modular width, through potential maximal cliques}
\author{Fedor V. Fomin\inst{1} \and Mathieu Liedloff\inst{2} \and Pedro Montealegre\inst{2} \and Ioan Todinca\inst{2}}

\institute{
Department of Informatics, University of Bergen, N-5020 Bergen, Norway,\\ \texttt{fedor.fomin@ii.uib.no}
\and
Univ. Orl\'{e}ans, INSA Centre Val de Loire, LIFO EA 4022, BP 6759, F-45067 Orl{\'e}ans Cedex 2, France,\\ \texttt{(mathieu.liedloff $\mid$ ioan.todinca $\mid$ pedro.montealegre)@univ-orleans.fr}
}

\date{\today}
\newcommand{\cF}{\mathcal{F}}
\newcommand{\cO}{\mathcal{O}}
\newcommand{\vc}{\operatorname{vc}}
\newcommand{\tw}{\operatorname{tw}}
\newcommand{\cw}{\operatorname{cw}}
\newcommand{\mw}{\operatorname{mw}}

\newcommand{\spmc}{\#\operatorname{pmc}}
\newcommand{\sm}{\setminus}
\newcommand{\pmc}{potential maximal clique}
\newcommand{\cmsot}{\operatorname{CMSO}_2}
\newcommand{\cmsoo}{\operatorname{CMSO}_1}
\newcommand{\msot}{\operatorname{MSO}_2}
\newcommand{\msoo}{\operatorname{MSO}_1}
\newcommand{\goldratio}{1.6181}
\newcommand{\pmcb}{1.7347}

\newcommand{\defproblem}[3]{
  \vspace{1mm}
\noindent\fbox{
  \begin{minipage}{0.96\textwidth}
  \begin{tabular*}{\textwidth}{@{\extracolsep{\fill}}lr} #1   \\ \end{tabular*}
  {\bf{Input:}} #2  \\
  {\bf{Task:}} #3
  \end{minipage}
  }
  \vspace{1mm}
}

\newcommand{\amims}{\textsc{Maximum Induced Subgraph with $\leq \ell$ copies of $\cF_m$-cycles}}

\newcommand{\mislkc}{\textsc{Maximum Induced Subgraph with $\leq \ell$ copies of $p$-cycles}}

\newcommand{\mislf}{{\sc Maximum Ind. Subgraph with   $ \leq  \ell$ copies of Minor Models from $\mathcal{F}$}}

\newcommand{\kip}{{\sc $k$-in-a-Path}}
\newcommand{\kit}{{\sc $k$-in-a-Tree}}
\newcommand{\kic}{{\sc $k$-in-a-Cycle}}
\newcommand{\kig}{{\sc $k$-in-a-Graph From $\mathcal{G}(t,\varphi)$}}

\newcommand{\igp}{{\sc  Independent $\mathcal{G}(t,\varphi)$-Packing}}

\newcommand{\fd}{{\sc Minimum $\mathcal{F}$-Deletion}}

\makeatletter
\def\imod#1{\allowbreak\mkern10mu({\operator@font mod}\,\,#1)}
\makeatother

\newcommand{\msphit}{\textsc{Max\- Induced\- Subgraph\- of\- $\tw\leq t$\- satisfiying\- $\varphi$}}

\newcommand{\npmc}{4} 

\begin{document}
\maketitle

\begin{abstract}
In this paper we give upper bounds on the number of \emph{minimal separators} and \emph{potential maximal cliques of graphs} w.r.t. two graph parameters, namely \emph{vertex cover} ($\vc$) and \emph{modular width} ($\mw$). We prove that for any graph, the number of minimal separators is $\cO^*(3^{\vc})$ and $\cO^*(\goldratio^{\mw})$, and the number of \pmc s is $\cO^*(4^{\vc})$ and $\cO^*(\pmcb^{\mw})$, and these objects can be listed within the same running times. (The $\cO^*$ notation suppresses polynomial factors in the size of the input.) Combined with known results~\cite{BoTo01,FoToVi14}, we deduce that a large family of problems, e.g., \textsc{Treewidth}, \textsc{Minimum Fill-in}, \textsc{Longest Induced Path}, \textsc{Feedback vertex set} and many others, can be solved in time $\cO^*(4^{\vc})$ or $\cO^*(\pmcb^{\mw})$. \end{abstract}

\section{Introduction}

The \emph{vertex cover} of a graph $G$, denoted by $\vc(G)$, is the minimum number of vertices that cover all edges of the graph.
The \emph{modular width} $\mw(G)$ can be defined as the maximum degree of a prime node in the modular decomposition of $G$ (see~\cite{TCHP08}  and Section~\ref{se:mw} for definitions).
The main results of this paper are of combinatorial nature: we show that the number of \emph{minimal separators} and the number of \emph{potential maximal cliques} of a graph are upper bounded by a function in each of these parameters. More specifically, we prove the number of minimal separators is at most $3^{vc}$ and $\cO^*(\goldratio^{\mw})$, and the number of potential maximal cliques is $\cO^*(4^{vc})$ and $\cO^*(\pmcb^{\mw})$, and these objects can be listed within the same running time bounds. Recall that the $\cO^*$ notation suppresses polynomial factors in the size of the input, i.e., $\cO^*(f(k))$ should be read as $f(k)\cdot n^{\cO(1)}$ where $n$ is the number of vertices of the input graph. Minimal separators and potential maximal cliques have been used for solving several classical optimization problems, e.g., \textsc{Treewidth}, \textsc{Minimum Fill-In}~\cite{FKTV08}, \textsc{Longest Induced Path}, \textsc{Feedback Vertex Set} or \textsc{Independent Cycle Packing}~\cite{FoToVi14}. Pipelined with our combinatorial bounds, we obtain a series of algorithmic consequences in the area of FPT algorithms parameterized by the vertex cover and the modular  width of the input graph. In particular, the problems mentioned above can be solved in time $\cO^*(4^{vc})$ and $\cO^*(\pmcb^{\mw})$. These results are complementary in the sense that graphs with small vertex cover are sparse, while graphs with small modular width may be dense.

Vertex cover and modular width are strongly related to treewidth ($\tw$) and cliquewidth ($\cw$) parameters, since for any graph $G$ we have $\tw(G) \leq \vc(G)$ and $\cw(G) \leq \mw(G)+2$. 
%
The celebrated theorem of Courcelle'~\cite{Courcelle90} states that 
all problems expressible in Counting Monadic Second Order Logic ($\cmsot$) can be solved in time $f(\tw) \cdot n$ for some function $f$ depending on the problem. A similar result for cliquewidth~\cite{CMR00} shows that all $\cmsoo$ problems can be solved in time $f(\cw) \cdot n$, if the clique-decomposition is also given as part of the input. (See the Appendix~\ref{ap:logic} for definitions of different types of logic. Informally, $\cmsot$ allows logic formulae with quantifiers over vertices, edges, edge sets and vertex sets, and counting modulo constants. The $\cmsoo$ formulae are more restricted, we are not allowed to quantify over edge sets.)

Typically function $f$ is a tower of exponentials, and the height of the tower depends on the formula. Moreover Frick and Grohe~\cite{FrGr04} proved that this dependency on treewidth or cliquewidth cannot be significantly improved in general. 
Lampis~\cite{Lampis12} shows that the running time for $\cmsot$ problems can be improved $2^{2^{\cO(\vc)}} \cdot n$ when parametrized by vertex cover, but he also shows that this cannot be improved to $\cO^*(2^{2^{o(\vc)}})$ (under the exponential time hypothesis). We are not aware of similar improvements for parameter modular width, but we refer to~\cite{GLO13} for discussions on problems parameterized by modular width. 

Most of our algorithmic applications concern a restricted, though still large subset of $\cmsot$ problems, but we guarantee algorithms that are single exponential in the vertex cover: $\cO^*(4^{vc})$ and in the modular width: $\cO^*(\pmcb^{\mw})$. We point out that our result for modular width extends the result of~\cite{FoVi10,FoToVi14}, who show a similar bound of $\cO^*(\pmcb^n)$ for the number of \pmc s and for the running times for these problems, but parameterized by the number of vertices of the input graph. 

We use the following  generic problem proposed by~\cite{FoToVi14}, that encompasses many classical optimization problems.
Fix an integer $t\geq 0$ and a $\cmsot$ formula $\varphi$. 
Consider the problem of finding, in the input graph $G$, an induced subgraph $G[F]$ together with a vertex subset $X \subseteq F$, such that the treewidth of $G[F]$ is at most $t$, the graph $G[F]$ together with the vertex subset $X$ satisfy formula $\varphi$, and $X$ is of maximum size under this conditions. This optimization problem is called \msphit:

\vspace{-0.2cm}
\begin{equation}\label{eq:opt_phi} 
\begin{array}{ll}
\mbox{Max}  &  |X|   \\
\mbox{subject to } &  \mbox{There is a set }   F\subseteq V    \mbox{ such that }  X\subseteq F;       \\
 &  \mbox{The treewidth of   }  G[F]    \mbox{ is at most }  t ;     \\
 &  (G[F],X)\models\varphi.  
\end{array}
\end{equation}
\vspace{-0.2cm}

Note that our formula $\varphi$ has a free variable corresponding to the vertex subset $X$. 
 For several examples, in  formula $\varphi$ the vertex set $X$ is actually equal to $F$. E.g., even when $\varphi$ only states that $X = F$, for $t=0$ we obtain the \textsc{Maximum Independent set problem}, and for $t=1$ we obtain the \textsc{Maximum Induced Forest}. If $t=1$ and  $\varphi$ states that $X=F$ and $G[F]$ is a path we obtain the \textsc{Longest Induced Path} problem. Still under the assumption that $X=F$, we can express the problem of finding the largest induced subgraph $G[F]$ excluding a fixed planar graph $H$ as a minor, or the largest induced subgraph with no cycles of length $0 \mod l$. But $X$ can correspond to other parameters, e.g. we can choose the formula $\varphi$ such that $|X|$ is the number of connected components of $G[F]$. Based on this we can express problems like \textsc{Independent Cycle Packing}, where the goal is to find an induced subgraph with a maximum number of components, and such that each component induces a cycle. 

The result of~\cite{FoToVi14} states that problem \msphit\ can be solved in a running time of the type $\spmc \cdot n^{t+4} \cdot f(\varphi,t)$ where $\spmc$ is the number of potential maximal cliques of the graph, assuming that the set of all potential maximal cliques is also part of the input. 
%
Thanks to our combinatorial bounds we deduce that the problem \msphit\  can be solved in time $\cO(\npmc^{\vc}n^{t+c})$ and  $\cO(\pmcb^{\mw}n^{t+c})$, for some small constant $c$.

There are several other graph parameters that can be computed in time $\cO^*(\spmc)$ if the input graph is given together with the set of its potential maximal cliques. E.g.,\textsc{Treewidth}, \textsc{Minimum Fill-in}~\cite{FKTV08}, their weighted versions~\cite{BoFo05,Gysel13} and several problems related to phylogeny~\cite{Gysel13},  or  \textsc{Treelength}~\cite{Lokshtanov10}. Pipelined with our main combinatorial result, we deduce that all these problems can be solved in time $\cO^*(\npmc^{\vc})$ or $\cO^*(\pmcb^{\mw})$. Recently Chapelle et al.~\cite{CLTV13} provided an algorithm solving \textsc{Treewidth} and \textsc{Pathwidth} in  $\cO^*(3^{\vc})$, 
but those completely different techniques do not seem to work for \textsc{Minimum Fill-in} or \textsc{Treelength}. The interested reader may also refer., e.g., to~\cite{CLP+14,FLM+08} for more (layout) problems parameterized by vertex cover.


\section{Minimal separators and \pmc s}\label{se:prelim}

Let $G=(V,E)$ be an undirected, simple graph. We denote by $n$ its number of vertices and by $m$ its number of edges. The \emph{neighborhood} of a vertex $v$ is
$N(v)=\{u\in V:~\{u,v\}\in E\}$. We say that a vertex $x$ \emph{sees} a vertex subset $S$ (or vice-versa) if $N(x)$ intersects $S$.
For a vertex set $S\subseteq V$ we denote by $N(S)$ the set $\bigcup_{v \in S} N(v)\setminus S$. We write $N[S]$ (resp. $N[x]$) for $N(S) \cup S$ (resp. $N(x) \cup \{x\}$). 
Also $G[S]$ denotes the subgraph of $G$ induced by $S$, and $G - S$ is the graph $G[V \setminus S]$.

A \emph{connected component} of graph $G$ is the vertex set of a maximal induced connected subgraph of $G$. Consider a vertex subset $S$ of graph $G$. Given two vertices $u$ and $v$, we say that $S$ is a $u,v$-separator if $u$ and $v$ are in different connected components of $G - S$. Moreover, if $S$ is inclusion-minimal among all $u,v$-separators, we say that $S$ is a \emph{minimal $u,v$-separator}. A vertex subset $S$ is called a \emph{minimal separator} of $G$ if $S$ is a $u,v$-minimal separator for some pair of vertices $u$ and $v$.

 Let $C$ be a component of $G - S$. If $N(C) = S$, we say that $C$ is a \emph{full component} associated to $S$. 

\begin{proposition}[folklore]\label{pr:full}
A vertex subset $S$ of $G$ is a minimal separator if $G - S$ has at least two full components associated to $S$. Moreover, $S$ is a minimal minimal $x,y$-separator if and only if $x$ and $y$ are in different full components associated to $S$. 
\end{proposition}

%

A graph $H$ is \emph{chordal} or \emph{triangulated} if every cycle with four or more vertices has a chord, i.e., an edge between two non-consecutive vertices of the cycle.
A {\em triangulation} of
a graph $G=(V,E)$ is a chordal graph $H = (V, E')$ such that $E
\subseteq E'$. Graph $H$ is a {\em minimal triangulation} of $G$ if
 for every
edge set $E''$ with $E \subseteq E'' \subset E'$, the
graph $F=(V, E'')$ is not chordal.



A set of vertices $\Omega \subseteq V$ of a graph $G$ is called a
{\em potential maximal clique} if there is a minimal triangulation
$H$ of $G$ such that $\Omega$ is a maximal clique of $H$. 


The following statement due to Bouchitt{\'e} and Todinca~\cite{BoTo01} provides a characterization of potential maximal cliques, and in particular allows to test in polynomial time if a vertex subset $\Omega$ is a potential maximal clique of $G$:

\begin{proposition}[\cite{BoTo01}]\label{pr:pmc_sep}
Let $\Omega \subseteq V$ be a set of vertices of the graph $G=(V,E)$ and
  $ \{ C_1, \ldots, C_p\}$ be the set of 
connected components of $G - \Omega$. We denote  ${\mathcal S}(\Omega) = \{
S_1,S_2, \ldots  , S_p\}$, where $S_i = N(C_i)$ for all $i \in \{1,\dots, p\}$. Then
$\Omega$ is a potential maximal clique of $G$ if and only if
\begin{enumerate}
\item each $S_i \in {\mathcal{S}}(\Omega)$ is strictly contained in $\Omega$;
\item the graph on the vertex set $\Omega$ obtained from $G[\Omega]$ by
completing each $S_i \in {\mathcal{S}}(\Omega)$ into a clique is a
complete graph.
\end{enumerate}
Moreover, if $\Omega$ is a potential maximal clique, then
$\mathcal S(\Omega)$ is  the set of minimal separators of $G$ contained
in $\Omega$.
\end{proposition}

Another way of stating the second condition is that for any pair of vertices $u,v \in \Omega$, if they are not adjacent in $G$ then 
there is a component $C$ of $G - \Omega$ seeing both $x$ and $y$. 

\begin{figure}[h]
\label{fi:cubewaterm}
\begin{center}
\includegraphics[scale=0.35]{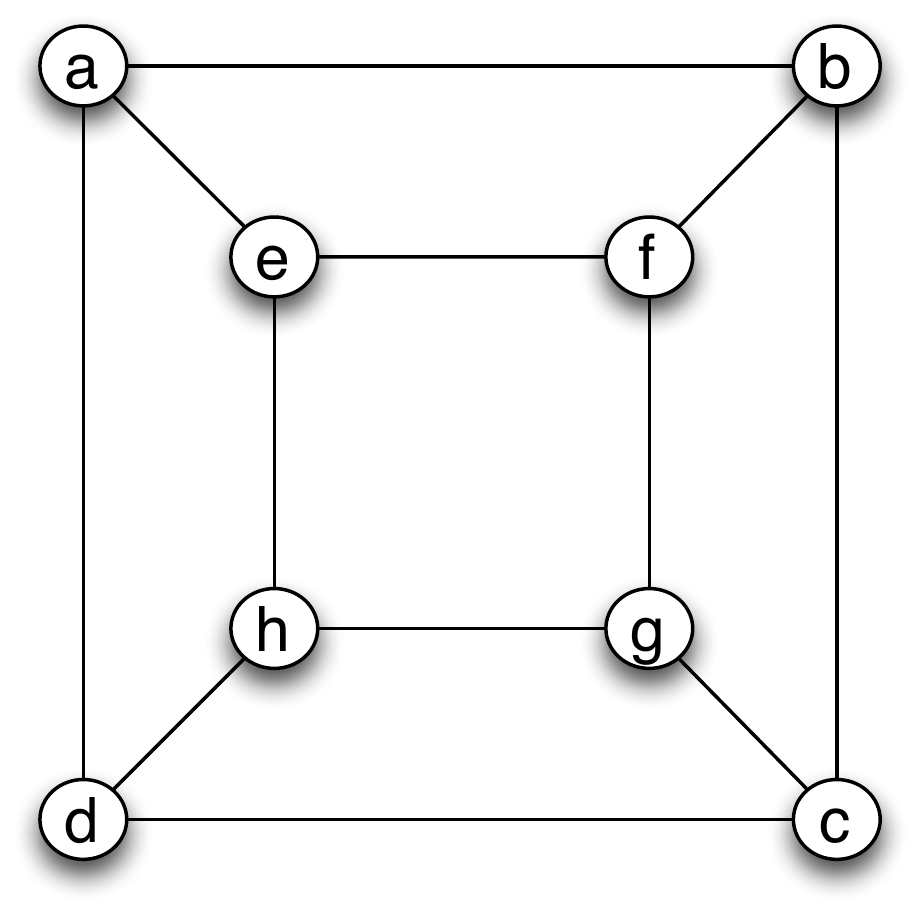}
\includegraphics[scale=0.4]{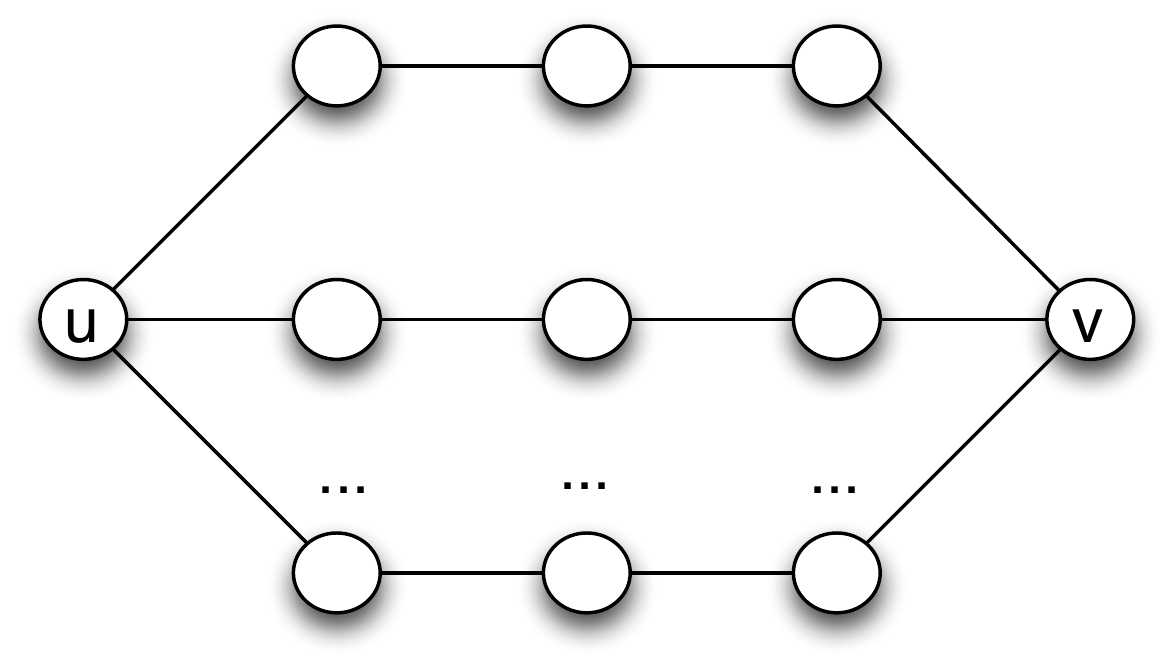}
\end{center}
\caption{Cube graph (left) and watermelon graph (right).}
\end{figure}

To illustrate Proposition~\ref{pr:pmc_sep}, consider, e.g., the cube graph depicted in Figure~\ref{fi:cubewaterm}. The set $\Omega_1 = \{a,e,g,c,h\}$ is a potential maximal clique and the minimal separators contained in $\Omega_1$ are $\{a,e,g,c\}$ and $\{a,h,c\}$. Another \pmc\ of the cube graph is $\Omega_2 = \{a,c,f,h\}$ containing the minimal separators $\{a,c,f\}$, $\{a,c,h\}$, $\{a,f,h\}$ and $\{c,f,h\}$.

Based on Propositions~\ref{pr:full} and~\ref{pr:pmc_sep}, one can easily deduce:
\begin{corollary}[see e.g., \cite{BoTo01}]\label{co:pmc_rec}\label{co:sep_rec}
There is an $O(m)$ time algorithm testing if a given vertex subset $S$ is a minimal separator of $G$, and $O(nm)$ time algorithm testing if a given vertex subset $\Omega$ is a potential maximal clique of $G$. 
\end{corollary}

We also need the following observation.
\begin{proposition}[\cite{BoTo01}]\label{pr:pmc_comp}
Let $\Omega$ be a \pmc\ of $G$ and let $S \subset \Omega$ be a minimal separator. Then $\Omega \setminus S$ is contained into a unique component $C$ of $G - S$, and moreover 
$C$ is a full component associated to $S$.
\end{proposition}

\section{Relations to vertex cover}\label{se:vc}
A vertex subset $W$ is a \emph{vertex cover} of $G$ if each edge has at least one endpoint in $W$. Note that if $W$ is a vertex cover, that $V \setminus W$ induces an \emph{independent set} in $G$, i.e. $G - W$ contains no edges. We denote by $\vc(G)$ the size of a minimum vertex cover of $G$. The parameter $\vc(G$) is called the \emph{vertex cover number} or simply (by a slight abuse of language) the \emph{vertex cover} of $G$.

\begin{proposition}[folklore]\label{pr:vc_FPT}
There is an algorithm computing the vertex cover of the input graph in time $\cO^*(2^{vc})$.
\end{proposition}

Let us show that any graph $G$ has at most $3^{\vc(G)}$ minimal separators.

\begin{lemma}\label{le:sep3part}
Let $G =(V,E)$ be a graph, $W$ be a vertex cover and $S \subseteq V$ be a minimal separator of $G$. Consider a three-partition $(D_1,S,D_2)$ of $V$ such that 
both $D_1$ and $D_2$ are formed by a union of components of $G - S$, and both $D_1$ and $D_2$ contain some full component associated to $S$.
Denote $D_1^W = D_1 \cap W$ and $D_2^W = D_2 \cap W$. 

Then $S \sm W  = \{ x \in V \setminus W \mid N(x) \text{~intersects both~} D_1^W \text{~and~} D_2^W\}$.
\end{lemma}
\begin{proof}
Let $C_1 \subseteq D_1$ and $C_2 \subseteq D_2$ be two full components associated to $S$.
Let $x \in S \setminus W$. Vertex $x$ must have neighbors both in $C_1$ and $C_2$, hence both in $D_1$ and $D_2$. Since $x \not\in W$ and $W$ is a vertex cover, we have $N(x) \subseteq W$. Consequently $x$ has neighbors both in  $D_1^W$ and $D_2^W$.

Conversely, let $x \in V \setminus W$ s.t. $N(x)$ intersects both $D_1^W$ and $D_2^W$. We prove that $x \in S$. By contradiction, assume that $x \not\in S$, thus $x$ is in some component $C$ of $G - S$. Suppose w.l.o.g. that $C \subseteq D_1$. Since $N(x) \subseteq C \cup N(C)$, we must have $N(x) \subseteq D_1 \cup S$. Thus $N(x)$ cannot intersect $D_2$---a contradiction.
\qed
\end{proof}

\begin{theorem}\label{th:sep_vc}
Any graph $G$ has at most $3^{\vc(G)}$ minimal separators. Moreover the set of its minimal separators can be listed in $\cO^*(3^{vc(G)})$ time.
\end{theorem}
\begin{proof}
Let $W$ be a minimum size vertex cover of $G$. For each three-partition $(D_1^W, S^W, D_2^W)$ of $W$, let $S = S^W \cup \{ x \in V \setminus W \mid N(x) \text{~intersects~} D_1^W \text{~and~} D_2^W\}$. According to Lemma~\ref{le:sep3part}, each minimal separator of $G$ will be generated this way, by an appropriate partition $(D_1^W, S^W, D_2^W)$ of $W$. Thus the number of minimal separators is at most $3^{\vc(G)}$, the number of three-partitions of $W$.

These arguments can be easily turned into an enumeration algorithm, we simply need to compute an optimum vertex cover then test, for each set $S$ generated from a three-partition, if $S$ is indeed a minimal separator. The former part takes $\cO^*(2^{(vc(G)})$ time by Proposition~\ref{pr:vc_FPT}, and the latter takes polynomial time for each set $S$ using Corollary~\ref{co:sep_rec}.
\qed
\end{proof}
 
Observe that the bound of Theorem~\ref{th:sep_vc} is tight up to a constant factor. Indeed consider the watermelon graph $W_{k,3}$  formed by $k$ disjoint paths of three vertices plus two vertices $u$ and $v$ adjacent to the left, respectively right ends of the paths (see Figure~\ref{fi:cubewaterm}). Note that this graph has vertex cover $k+2$ (the minimum vertex cover contains the middle of each path and vertices $u$ and $v$) and it also has $3^{k}$ minimal $u,v$-separators, obtained by choosing arbitrarily one of the three vertices on each of the $k$ paths. 

 
We now extend Theorem~\ref{th:sep_vc} to a similar result on \pmc s. Let us distinguish a particular family of potential maximal cliques, which have \emph{active} separators. They have a particular structure which makes them easier to handle.

\begin{definition}[\cite{BoTo02}]\label{de:active}
Let $\Omega \subseteq V$ be a potential maximal clique of graph $G=(V,E)$, let 
  $ \{ C_1, \ldots, C_p\}$ be the set of 
connected components of $G -  \Omega$ and let $S_i = N(C_i)$, for $1 \leq i \leq p$.

Consider now the graph $G^+$ obtained from $G$ by completing into a clique all minimal separators $S_j$, $2 \leq i \leq p$, such that $S_j \not\subseteq S_1$.

We say that $S_1$ is an \emph{active separator} for $\Omega$ if $\Omega$  is not a clique in this graph $G^+$. A pair of vertices $x,y \in \Omega$ that are not adjacent in $G^+$ is called an \emph{active pair}. Note that, by Proposition~\ref{pr:pmc_sep}, we must have $x,y \in S_1$. 
\end{definition}

The following statement characterizes potential maximal cliques with active separators.

\begin{proposition}\label{pr:pmc_active}
Let $\Omega$ be a potential maximal clique having an active separator $S \subset \Omega$, with an active pair $x,y \in S$. Denote by $C$ the unique component of $G - S$ containing $\Omega \setminus S$. Then $\Omega \setminus S$ is a minimal $x,y$-separator in the graph $G[C \cup \{x,y\}]$. 
\end{proposition}

Again on the cube graph of Figure~\ref{fi:cubewaterm}, for the \pmc\ $\Omega_1 = \{ a,e,g,c,h\}$, both minimal separators are active. E.g., for the minimal separator $S = \{a,e,g,c\}$ the pair $\{e,g\}$ is active. Not all potential maximal cliques have active separators, as illustrated by the potential maximal clique $\Omega_2 = \{a,c,f,h\}$ of the same graph.

Let us first focus on potential maximal cliques having an active separator. 
We give a result similar to Lemma~\ref{le:sep3part}, showing that such a potential maximal clique can be determined by a certain partition of the vertex cover $W$ of $G$.
 
 \begin{lemma}\label{le:pmc4part}
Let $G =(V,E)$ be a graph and $W$ be a vertex cover of $G$. Consider a \pmc\ $\Omega$ of $G$ having an active separator $S \subseteq \Omega$ and an active pair $x,y \in S$.
Let $C$ be the unique connected component of $G - S$ intersecting $\Omega$ and let $D_S$ be the union of all other connected components of $G - S$. 
Denote by $D_x$  the union of components of $G - \Omega$ contained in $C$, seeing $x$, by $D_y$ the union of components of $G - \Omega$ contained in $C$ not seeing $x$. 

Now let $D_S^W = D_S \cap W$, $D_x^W = D_s \cap W$ and $D_y^W = D_y \cap W$.  

Then one of the following holds:
\begin{enumerate}
\item\label{it:1} There is a vertex $t \in \Omega$ such that $\Omega \setminus S = N(t) \cap C$.
\item There is a vertex $t \in \Omega$ such that $\Omega = N[t]$.
\item A vertex $z \not\in W$ is in $\Omega$ if and only if
\begin{enumerate}
\item\label{it:SW} $z$ sees $D_S^W$ and $D_x^W \cup D_y^W$, or
\item\label{it:TW} $z$ does not see $D_S^W$ but is sees  $D_x^W \cup \{x\}$, $D_y^W \cup \{y\}$ and $D_x^W \cup D_y^W$.
\end{enumerate}
\end{enumerate}
\end{lemma} 
\begin{proof}

Note that $D_x, D_y, D_S$ and $\Omega$ form a partition of the vertex set $V$.

We first prove that any vertex $z \not\in W$ satisfying conditions~\ref{it:SW} or~\ref{it:TW} must be in $\Omega$. 

Consider first the case~\ref{it:SW} when $z$ sees $D_S^W$ and $D_x^W \cup D_y^W$. So $z$ sees $D_S$ and $C$; we can apply Lemma~\ref{le:sep3part} to partition $(D_S,S,C)$ 
thus $z \in S$.
Consider now the case~\ref{it:TW} when $z$ sees $D_x^W \cup D_y^W$, $D_x\cup \{x\}$ and $D_y \cup \{y\}$ but not $D_S^W$. Again by Lemma~\ref{le:sep3part} applied to partition $(D_S,S,C)$, vertex $z$ cannot be in $S$. Since $z$ has a neighbor in $D_x \cup D_y$, we have $z \in C$. Let $H = G[C \cup \{x,y\}]$ and $T = \Omega \cap C$ (thus we also have $T = \Omega \setminus S$). Recall that $T$ is an $x,y$-minimal separator in $H$ by Proposition~\ref{pr:pmc_active}. By definition of set $D_x$, we have that $D_x \cup \{x\}$ is exactly the component of $H - T$ containing $x$. Note that $D_y \cup \{y\}$ is the union of the component of $H - T$  containing $y$ and of all other components of $H - T$ (that no not see $x$ nor $y$). By applying Lemma~\ref{le:sep3part} on graph $H$, with vertex cover $(W \cap C) \cup \{x,y\}$ and with partition $(D_x \cup \{x\}, T, D_y \cup \{y\})$ we deduce that $z \in T$. 

Conversely, let $z \in \Omega \sm W$. We must prove that either $z$ satisfies conditions~\ref{it:SW} or~\ref{it:TW}, or we are in one of the first two cases of the Lemma. 
We distinguish the cases $z \in S$ and $z \in T$. When $z \in S$, by Lemma~\ref{le:sep3part} applied to partition $(D_S,S,C)$, $z$ must see $D_S$ and $C$. If $z$ sees some vertex in $C \setminus \Omega$, we are done because $z$ sees  $D_x^W \cup D_y^W$ so we are in case~\ref{it:SW}. Assume now that $N(z) \cap C \subseteq \Omega$, we prove that actually $N(z) \cap C = \Omega \cap C = T$, so we are in case~\ref{it:1}. Assume there is $u \in T \setminus N(z)$. By Proposition~\ref{pr:pmc_sep}, there must be a connected component $D$ of $G - \Omega$ such that $z,u \in N(D)$. Since $u \in C$, this component $D$ must be a subset of $C$, so $D \subseteq C \sm \Omega$. Together with $z \in N(D)$, this contradicts the assumption $N(z) \cap C \subseteq \Omega$. 

It remains to treat the case $z \in T$. Clearly $z \in C$ cannot see $D_S$ because $S$ separates $C$ from $D_S$. 
We again take graph $H$, with vertex cover $(W \cap C) \cup \{x,y\}$, and apply Lemma~\ref{le:sep3part} with partition $(D_x \cup \{x\}, T, D_y \cup \{y\})$. We deduce that $z$ sees both  $D_x^W \cup \{x\}$ and
$D_y^W \cup \{y\}$. Assume that $z$ does not see $D_x^W \cup D_y^W$. So $N(z) \cap C \setminus \Omega = \emptyset$ thus $N[z] \subseteq \Omega$. 
If $\Omega$ contains some vertex $u \not\in N[z]$, no component of $G - \Omega$ can see both $z$ and $u$ (because $N(z) \subseteq \Omega$), contradicting Proposition~\ref{pr:pmc_sep}. We conclude that either $z$ sees $D_x^W \cup D_y^W$ (so satisfies condition~\ref{it:TW}) or $\Omega = N[z]$ (thus we are in the second case of the Lemma).
\qed
\end{proof}

\begin{theorem}\label{th:pmca_vc}
Any graph $G$ has $\cO^*(\npmc^{\vc(G)})$ \pmc s with active separators. Moreover the set of its \pmc s with active separators can be listed in $\cO^*(\npmc^{vc(G)})$ time.
\end{theorem}
\begin{proof}
The number of \pmc s with active separators satisfying the second condition of Lemma~\ref{le:pmc4part} is at most $n$, and they can all be listed in polynomial time by checking, for each vertex $t$, if $N[t]$ is a \pmc.

For enumerating the \pmc s with active separators satisfying the first condition of Lemma~\ref{le:pmc4part}, we enumerate all minimal separators $S$ using Theorem~\ref{th:sep_vc}, then 
for each $t \in S$ and each of the at most $n$ components $C$ of $G - S$ we check if $S \cup (C \cap N(t))$ is a \pmc. Recall that testing if a vertex set is a \pmc\ can be done in polynomial 
time by Corollary~\ref{co:pmc_rec}. Thus the whole process takes $\cO^*(3^{vc(G)})$ time, and this is also an upper bound on the number of listed objects.

It remains to enumerate the \pmc s with active separators satisfying the third condition of Lemma~\ref{le:pmc4part}. For this purpose, we ``guess'' the sets $D_S^W$ $D_x^W$, $D_y^W$ as in the Lemma and then we compute $\Omega$. More formally, for each four-partition $(D_S^W, D_x^W, D_y^W, \Omega^W)$ of $W$, we let $\Omega^{\overline{W}}$ be the set of vertices $z \not\in W$ satisfying conditions~\ref{it:SW} or~\ref{it:TW} of Lemma~\ref{le:pmc4part}, and we test using Corollary~\ref{co:pmc_rec} if $\Omega  = \Omega^W \cup \Omega^{\overline{W}}$ is indeed a \pmc. By Lemma~\ref{le:pmc4part}, this enumerates in $\cO^*(\npmc^{vc(G)})$ all \pmc s of this type. 
\qed
\end{proof}

For counting and enumerating all \pmc s of graph $G=(V,E)$, including the ones with no active separators, we apply the same ideas as in~\cite{BoTo02}, based on the following statement.
%

\begin{proposition}[\cite{BoTo02}]\label{pr:pmc_a}
Let $G=(V,E)$ be a graph, let $u$ be an arbitrary vertex of $G$ and $\Omega$ be a potential maximal clique of $G$. Denote by $G - u$ the graph $G[V \sm\{u\}]$. Then one of the following holds.
\begin{enumerate}
\item\label{it:active} $\Omega$ has an active minimal separator $S$.
\item\label{it:witha} $\Omega$ is a potential maximal clique of $G - u$.
\item\label{it:minusa} $\Omega \setminus \{u\}$ is a potential maximal clique of $G - u$.
\item\label{it:minsep} $\Omega \setminus \{u\}$ is a minimal separator of $G$.
\end{enumerate}
\end{proposition}

\begin{theorem}\label{th:pmc_vc}
Any graph $G$ has $\cO^*(\npmc^{\vc(G)})$ \pmc s. Moreover the set of its \pmc s can be listed in $\cO^*(\npmc^{vc(G)})$ time.
\end{theorem}
\begin{proof}
Let $(v_1,\dots, v_n)$ be an arbitrary ordering of the vertices of $V$. Denote by $G_i$ the graph $G[\{v_1, \dots, v_i\}]$ induced by the first $i$ 
vertices, for all $i, 1 \leq i \leq n$. Let $k =  \vc(G)$. Note that for all $i$ we have $\vc(G_i) \leq k$. Actually, if $W$ is a vertex cover of $G$, then $W_i = W \cap \{v_1, \dots, v_i\}$ is a vertex 
cover of $G_i$. In particular, by Theorems~\ref{th:sep_vc} and~\ref{th:pmca_vc}, each $G_i$ has at most $3^k$ minimal separators and $\cO^*(\npmc^k)$ \pmc s with active separators.

For $i = 1$, graph $G_1$ has a unique \pmc\ equal to $\{v_1\}$.

For each $i$ from $2$ to $n$, in increasing order, we compute the \pmc s of $G_i$ from those of $G_{i-1}$ using Proposition~\ref{pr:pmc_a}. 
Observe that $G_{i-1} = G_i - v_i$. We initialize the set of \pmc s of $G_i$ with the ones having active separators. This can be done in $\cO^*(\npmc^k)$ time by Theorem~\ref{th:pmca_vc}.
Then for each minimal separator $S$ of $G_i$ we check if $\Omega = S \cup \{v_i\}$ is a \pmc\ of $G_i$ and if so we add it to the set. 
This takes $\cO^*(3^k)$ time by Theorem~\ref{th:sep_vc} and Corollary~\ref{co:pmc_rec}.
Eventually, for each \pmc\ $\Omega'$ of $G_{i-1}$, we test using Corollary~\ref{co:pmc_rec} if $\Omega'$ (resp. $\Omega' \cup \{v_i\}$) is a \pmc\ of $G_i$. If so, we add it to the set of \pmc s of $G_i$. The running time of this last part is the number of \pmc s of $G_{i-1}$ times $nm$. Altogether, it takes $\cO^*(\npmc^k)$ time.

By Proposition~\ref{pr:pmc_a}, this algorithm  covers alls cases and thus lists all \pmc s of $G_i$. Hence for $i=n$ we obtain all \pmc s of $G$, and they have been enumerated in $\cO^*(\npmc^k)$ time.
\qed
\end{proof}

\section{Relations to modular width}\label{se:mw}

A \emph{module} of graph $G = (V,E)$ is a set of vertices $W$ such that, for any vertex $x \in V \setminus W$, either $W \subseteq N(x)$ or $W $ does not intersect $N(x)$. For the reader familiar with the modular decompositions of graphs, the modular width $\mw(G)$ of a graph $G$ is the maximum size of a prime node in the modular decomposition tree. Equivalently,
graph $G$ is of modular width at most $k$ if:
\begin{enumerate}
\item $G$ has at most one vertex (the base case).
\item $G$ is a disjoint union of graphs of modular width at most $k$.
\item $G$ is a \emph{join} of graphs of modular width at most $k$. I.e., $G$ is obtained from a family of disjoint graphs of modular width at most $k$ by taking the disjoint union and then adding all possible edges between these graphs.
\item the vertex set of $G$ can be partitioned into $p \leq k$ modules $V_1, \dots, V_p$ such that $G[V_i]$ is of modular width at most $k$, for all $i, 1 \leq i \leq p$.
\end{enumerate}
The modular width of a graph can be computed in linear time, using e.g.~\cite{TCHP08}. Moreover, this algorithm outputs the algebraic expression of $G$ corresponding to this grammar. 

Let $G=(V,E)$ be a graph with vertex set $V=\{v_1, \dots, v_k\}$ and let $M_i = (V_i,E_i)$ be a family of pairwise disjoint graphs, for all $i$, $1 \leq i \leq k$. Denote by $H$ the graph obtained from $G$ by replacing each vertex $v_i$ by the module $M_i$. I.e., $H=(V_1 \cup \dots \cup V_k, E_1 \cup \dots \cup E_k \cup \{ab \mid a \in V_i, b \in V_j \text{~s.t.~} v_iv_j \in E\})$. We say that graph $H$ has been obtained from $G$ by \emph{expanding} each vertex $v_i$ by the module $M_i$.

A vertex subset $W$ of $H$ is an \emph{expansion} of vertex subset $W_G$ of $G$ if $W = \cup_{v_i \in W_G} V_i$. Given a vertex subset $W$ of $H$, the \emph{contraction} of $W$ is $\{v_i \mid V_i \text{~intersects~} W\}$. 

\begin{lemma}\label{le:sepmod1}
Let $S$ be a minimal $y,z$-separator of $H$, for $y,z \in V_i$. Then $S \cap V_i$ is a minimal separator of $M_i$ and $S \setminus V_i = N_H(V_i)$. 
\end{lemma}
\begin{proof}
Note that all vertices of $N_H(V_i)$ are in $N_H(y) \cap N_H(z)$, by construction of graph $H$ and the fact that $y$ and $z$ are in the same module induced by $V_i$. Therefore $N_H(V_i)$ must be contained in $S$. Let $S_i = S \cap V_i$. Since $H[V_i] = M_i$, we have that $S_i$ separates $z$ and $y$ in graph $M_i$. Assume that $S_i$ is not a minimal $y,z$-separator of $M_i$, so let $S'_i \subsetneq S_i$ be a minimal $y,z$-separator in graph $M_i$. We claim that $S'_i \cup N_H(V_i)$ is a $y,z$-separator in $H$. Indeed each $y,z$-path of $H$ is either contained in $V_i$ (in which case it intersects $S'_i$) or intersects $N_H(V_i)$. In both cases, it passes through $S'_i \cup N_H(V_i)$, which proves the claim. Since $S'_i \cup N_H(V_i)$ is a subset of $S$ and $S$ is a $y,z$-minimal separator of $H$, the only possibility is that $S = S'_i \cup N_H(V_i)$. This proves that $S \cap V_i$ is a minimal separator of 
$M_i$ and $S \setminus V_i = N_H(V_i)$.
\qed
\end{proof}

\begin{lemma}\label{le:sepmod}
Let $S$ be a minimal separator of $H$. Assume that some $V_i$ intersects $S$, but is not contained in $S$.
Then $V_i$ intersects all full components of $H - S$ associated to $S$. In particular $S \cap V_i$ is a minimal separator in $M_i$ and 
$S \setminus V_i = N_H(V_i)$.
\end{lemma}
\begin{proof}
Let $x \in V_i \cap S$ and $t \in V_i \setminus S$. 
By Proposition~\ref{pr:full}, there are at least two full  components of $H - S$, associated to $S$. Let $C$ be one of them, not containing $t$. Let $z$ be a neighbor of $x$ in $C$, we prove that $z \in V_i$. If $z \not\in V_i$, then $z \in N_H(V_i)$, and since $V_i$ is a module in $H$ we also have $z \in N_H(t)$. This contradicts the fact that $t$ and $z$ are in different components of $H - S$. It remains that  $z \in V_i$. By applying the same argument for $z$ instead of $t$, it follows that $V_i$ intersects each full component $D$ of $H - S$ and moreover $x$ has a neighbor in $D \cap V_i$. 

By Proposition~\ref{pr:full}, $S$ is a minimal $y,z$-separator in $H$, for some $y, z \in N_H(x) \cap V_i$. The rest follows by Lemma~\ref{le:sepmod1}.
\qed
\end{proof}

\begin{lemma}\label{le:sepexp}
Let $S$ be a minimal separator of $H$. One of the following holds~:
\begin{enumerate}
\item $S$ is the expansion of a minimal separator $S_G$ of $G$. 
\item There is $i \in \{1,\dots, k\}$ such that $S \cap V_i$ is a minimal separator of $M_i$ and $S \setminus V_i = N_H(V_i)$.
\end{enumerate}
\end{lemma}
\begin{proof}
Assume there is a set $V_i$ intersecting $S$ but not contained in it. By Lemma~\ref{le:sepmod}, $S \cap V_i$ is a minimal separator of $M_i$ and $S \setminus V_i = N_H(V_i)$. Hence we are in the second case of the Lemma.

Otherwise, for any $V_i$ intersecting $S$, we have $V_i \subseteq S$. Thus $S$ is the expansion of a vertex subset $S_G$ of $G$, formed exactly by the vertices $v_i$ of $G$ such that $V_i$ intersects $S$.  Let $C$ and $D$ be two full components of $H - S$ associated to $S$ and let $a \in C$, $b \in D$. Recall that, by Proposition~\ref{pr:full}, $S$ is a minimal $a,b$-separator of $H$. Let $V_k$ be the set containing $a$ and $V_l$ the set containing $b$. Consider first the possibility that $k = l$. Then, by Lemma~\ref{le:sepmod1}, $S$ satisfies the second condition of this lemma, for $i=k=l$. (This case may occur when $M_k$ is disconnected and $S = N_H(V_k)$.)

It remains the case $k \neq l$. We prove that $S_G$ is a minimal $v_k,v_l$-separator of $G$. Consider a $v_k,v_l$ path of $G$. If this path does not intersect $S_G$ in $G$, then there is a path from $a$ to $b$ in $H - S$, obtained by replacing each vertex $v_j$ of the path by some vertex of $V_j$ ($v_k$ and $v_l$ are replaced by $a$ and $b$ respectively). This would contradict the fact that $S$ separates $a$ and $b$ in $H$. Therefore $S_G$ is indeed a $v_k,v_l$-separator in $G$. Assume that $S_G$ is not minimal among the $v_k,v_l$-separators of $G$, and let $v_j \in S_G$ such that $S_G \setminus \{v_j\}$ separates $v_k$ and $v_l$ in $G$. We claim that $S \setminus V_j$ also separates $a$ from $b$ in $H$. By contradiction, assume there is a path from $a \in C \cap V_k$ to $b \in D \cap V_l$ in $H$, avoiding $S \setminus V_j$. By contracting, on this path, all vertices belonging to a same $V_i$ into vertex $v_i$, we obtain a path (or a connected subgraph) joining $v_k$ to $v_l$ in $G$. This contradicts the fact that all such paths should intersect $S_G \setminus \{v_j\}$. Therefore $S_G$ is a minimal separator of $G$. 
\qed
\end{proof}

Lemma~\ref{le:sepexp} provides an injective mapping from the set of minimal separators of $H$ to the union of the sets of minimal separators of $G$ and of the graphs $M_i$. Therefore we have:
\begin{corollary}\label{co:sepexp}
The number of minimal separators of $H$ is at most the number of minimal separators of $G$ plus the number of minimal separators of each $M_i$. 
\end{corollary}

We now aim to prove a statement equivalent of Corollary~\ref{co:sepexp}, for potential maximal cliques instead of minimal separators. 

\begin{lemma}\label{le:pmcexp}
Let $\Omega$ be a \pmc\ of $H$, and let $\Omega_G = \{v_i \mid V_i \text{~intersects~} \Omega\}$. Assume that $\Omega$ is the expansion of $\Omega_G$, i.e. $\Omega  = \cup_{v_i \in \Omega_G} V_i$. Then $\Omega_G$ is a \pmc\ of $G$.
\end{lemma} 
\begin{proof}
We prove that $\Omega_G$ satisfies, in graph $G$, the conditions of Proposition~\ref{pr:pmc_sep}. For the first condition, let $C_G$ be a component of $G - \Omega_G$ and let $S_G = N_G(C_G)$. Assume that $S_G$ is not strictly contained in $\Omega_G$, hence $S_G = \Omega_G$. Let $C$ be the expansion of $C_G$ in $H$ and note that $N_H(C)$ is the expansion of $N_G(C_G)$, thus $N_H(C) = \Omega$. If $C_G$ is formed by at least two vertices, since $G[C_G]$ is connected then so is $H[C]$.  Therefore, in graph $H$, we have $N_H(C) = \Omega$ and $C$ is a component of $H - \Omega$. But this contradicts the first condition of Proposition~\ref{pr:pmc_sep} applied to the \pmc\ $\Omega$ of $H$. In the case that $C_G$ is formed by a unique vertex $v_k$, its expansion $C = V_k$ might not induce a connected subset in $H$ (if $M_k$ is disconnected). But it is sufficient to consider a connected component $V'_k$ of $H[V_k]$, and again this is also a component of $H - \Omega$ with the property that its neighborhood in $H$ is the whole set $\Omega$, contradicting Proposition~\ref{pr:pmc_sep} applied to $\Omega$.

For the second condition of Proposition~\ref{pr:pmc_sep}, let $v_j, v_k \in \Omega_G$ such that $v_jv_k$ is not an edge of $G$. Let $a \in V_j$ and $b \in V_k$. These vertices are non-adjacent in $H$, so by Proposition~\ref{pr:pmc_sep} applied to the \pmc\ $\Omega$ of $H$ there must be a component $C$ of $H - \Omega$ seeing both $a$ and $b$. Consider an $a,b$-path in $H[C \cup \{a,b\}]$. The contraction of this path contains a $v_j,v_k$-path in $G$, whose internal vertices are not in $\Omega_G$. This proves that $v_j$ and $v_k$ are in the neighborhood of a same component of $G - \Omega_G$, thus $\Omega_G$ satisfies the second condition of Proposition~\ref{pr:pmc_sep}.
\qed
\end{proof}

\begin{lemma}\label{le:pmcmod}
Let $\Omega$ be a \pmc\ of $H$, and assume that there is some set $V_i$ that intersects $\Omega$ but is not contained in $\Omega$. Then $\Omega \cap V_i$ is a \pmc\ of $M_i$ and $\Omega \setminus V_i = N_H(V_i)$.
\end{lemma}
\begin{proof}
Let $V_i$ be a vertex set that intersects $\Omega$, but is not contained in $\Omega$. 

We claim that $\Omega$ contains a minimal $x,y$-separator of $H$, for some pair of vertices $x \in \Omega \cap V_i$ and $y \in V_i \setminus \Omega$. If $V_i$ intersects some minimal separator $S$ contained in $\Omega$, then by Lemma~\ref{le:sepmod}, $S \setminus V_i$ is a minimal separator of $M_i$ and $V_i$ intersects all full components of $H - S$ associated to $S$, which proves our claim. Consider the case when $V_i$ does not intersect any minimal separator of $H$ contained in $\Omega$. Let $x \in \Omega \cap V_i$ and note that $\Omega \setminus \{x\}$ separates, in graph $H$, vertex $x$ from all other vertices (because by Proposition~\ref{pr:pmc_sep}, $x$ has no neighbors in $H - \Omega$). Recall that $V_i \not\subseteq \Omega$, thus there is some $y \in V_i \setminus \Omega$, then $\Omega$ contains some minimal $x,y$-separator $S$ in graph $H$.

By Lemma~\ref{le:sepmod1}, $S \setminus V_i$ is a minimal separator of $M_i$ and $V_i$ intersects all full components of $H - S$ associated to $S$. Let $C$ be the unique component of $H - S$ intersecting $\Omega$; recall that it exists and moreover it is full w.r.t. $S$, by Proposition~\ref{pr:pmc_comp}. Then, by Lemma~\ref{le:sepmod}, $C$ also intersects $V_i$. Also by Lemma~\ref{le:sepmod}, $S \setminus V_i = N_H(V_i)$ and $S_i = S \cap V_i$ is a minimal separator of $M_i$. We claim that actually $C \subseteq V_i$ and $C$ is also a full component of $M_i - S_i$. Recall that $S \setminus V_i = N_H(V_i)$ separates in graph $H$ the vertices of $V_i$ from the rest of the graph. Since $C$ intersects $V_i$, $H[C]$ is connected and $N_H(V_i)$ separates $V_i$ from all other vertices, we must have $C \subseteq V_i$. Since $H[C]$ is connected, so is $M_i[C]$, thus $C$ is contained in some component of $M_i - S_i$. But each such component is also a component of $H - S$, hence $C$ is both a component of $H-S$ and of $M_i - S_i$. In particular $\Omega \cap C \subseteq V_i$.

It remains to prove that $\Omega_i = \Omega \cap V_i$  is a \pmc\ of $M_i$. By the above observations, we also have $\Omega_i = \Omega \setminus N_H(V_i)$. We show that $\Omega_i$ satisfies, in graph $M_i$, the conditions of Proposition~\ref{pr:pmc_sep}. Let $D$ be a component of $M_i - \Omega_i$. Observe that $D$ is also a component of $H - \Omega$ and let $T = N_{M_i}(D)$. Either $D$ is a component of $M_i - \Omega_i$ disjoint from $C$, or it is contained in $C$. In the former case, $T$ is a subset of $S_i$, hence $T$ is a strict subset of $\Omega_i$ (since $S_i$ is itself a strict subset of $\Omega_i$ by Proposition~\ref{pr:pmc_sep} applied to \pmc\ $\Omega$ of $H$). In the later case, if $T = \Omega_i$, note that $N_H(D) = \Omega$ because $\Omega \setminus V_i = N_H(V_i)$ is also contained in the neighborhood of $D$ in $H$. This contradicts Proposition~\ref{pr:pmc_sep} applied to \pmc\ $\Omega$ of $H$.

For the second condition, let $x,y \in \Omega_i$, non-adjacent in $M_i$. Then there is a component $F$ of $H - \Omega$ seeing, in graph $H$, both $x$ and $y$ (by Proposition~\ref{pr:pmc_sep} applied to $\Omega$). Since this component sees $V_i$, it must be contained in $V_i$. So $F$ is also a component of $M_i - \Omega_i$ seeing both $x$ and $y$ in $M_i$, which concludes our proof.
\qed
\end{proof}

From Lemmata~\ref{le:pmcexp} and~\ref{le:pmcmod}, we directly deduce~:
\begin{lemma}\label{le:pmcmw}
Let $\Omega$ be a \pmc\ of $H$. One of the following holds~:
\begin{enumerate}
\item $\Omega$ is the expansion of a \pmc\ $\Omega_G$ of $G$. 
\item There is some $i \in \{1, \dots, k\}$ such that $\Omega \cap V_i$ is a \pmc\ of $M_i$ and $\Omega \setminus V_i = N_H(V_i)$.
\end{enumerate}
\end{lemma}

The previous lemma provides an injective mapping from the set of \pmc s of $H$ to the union of the sets of \pmc s of $G$ and of the graphs $M_i$. Therefore we have:
\begin{corollary}\label{co:pmcexp}
The number of \pmc s of $H$ is at most the number of \pmc s of $G$ plus the number of \pmc s of each $M_i$. 
\end{corollary}

The following proposition bounds the number of minimal separators and \pmc s of arbitrary graphs with respect to $n$.
\begin{proposition}[\cite{FoVi10,FoVi12}]\label{pr:nbsep}\label{pr:nbpmc}
Every $n$-vertex graph has $\cO(\goldratio^n)$ minimal separators and $\cO(\pmcb^n)$ \pmc s. Moreover, these objects can be enumerated within the same running times. 
\end{proposition}

We can now prove the main result of this section.
\begin{theorem}\label{th:sepmw}\label{th:pmcmw}
For any graph $G=(V,E)$, the number of its minimal separators is $\cO(n \cdot \goldratio^{\mw(G)})$ and the number of its potential maximal cliques is $\cO(n \cdot \pmcb^{\mw(G)})$. Moreover, the minimal separators and the potential maximal cliques can be enumerated in time $\cO^*(\goldratio^{\mw(G)})$ and $\cO^*(\pmcb^{\mw(G)})$ time respectively.
\end{theorem}
\begin{proof}
Let $k = \mw(G)$. By definition of modular width, there is a decomposition tree of graph $G$, each node corresponding to a leaf, a disjoint union, a join or a decomposition into at most $k$ modules. The leaves of the decomposition tree are disjoint graphs with at single vertex, thus these vertices form a partition of $V$. In particular, there are at most $n$ leaves and, since each internal node is of degree at least two, there are $O(n)$ nodes in the decomposition tree. For each node $N$, let $G(N)$ be the graph associated to the subtree rooted in $N$. We prove that  $G(N)$ has $\cO(n(N) \cdot  \goldratio^k)$ minimal separators and $\cO(n(N) \cdot \pmcb^k)$ \pmc, where $n(N)$ is the number of nodes of the subtree rooted in $N$. We proceed by induction from bottom to top. The statement is clear when $N$ is a leaf.

Let $N$ be an internal node $N_1, N_2, \dots, N_p$ be its sons in the tree. 
Graph $G(N)$ is the expansion of some graph $G'(N)$ by replacing the $i$-th vertex with module $G(N_i)$. If $N$ is a \emph{join} node, then $G'(N)$ is a clique. When $N$ is a \emph{disjoint union} node, graph $G'(N)$ is an independent set, and in the last case $G'(N)$ is a graph of at most $k$ vertices. In all cases, by Proposition~\ref{pr:nbsep} graph $G'(N)$ has  $\cO(\goldratio^k)$ minimal separators. Thus $G(N)$ has at most $\cO(\goldratio^k)$ more minimal separators than all its sons taken together, which completes our proof for minimal separators. 

Concerning potential maximal cliques, when $G'(N)$ is a clique it has exactly one potential maximal clique, and when $G'(N)$ is of size at most $k$ is has $\cO(\pmcb^k)$ \pmc s. We must be more careful in the case when $G'(N)$ is an independent set (i.e., $N$ is a disjoint union node), since in this case it has $p$ \pmc s, one for each vertex, and $p$ can be as large as $n$. Consider a \pmc\ $\Omega$ of $G(N)$ corresponding to an expansion of vertices of $G'(N)$ (see Lemma~\ref{le:pmcmw}).  It follows that this \pmc\ is exactly the vertex set of some $G(N_i)$, for a child $N_i$ of $N$. By construction this vertex set is disconnected from the rest of $G(N)$, and by Proposition~\ref{pr:pmc_sep} the only possibility is that this vertex set induces a clique in $G(N)$. But in this case $\Omega$ is also a \pmc\ of $G(N_i)$. This proves that, when $N$ is of type disjoint union, $G(N)$ has no more \pmc s than the sum of the numbers of \pmc s of all its sons. We conclude that the whole graph $G$ has  $\cO(n \cdot \pmcb^k)$ \pmc s. All our arguments are constructive and can be turned directly into enumeration algorithms for these objects.
\qed
\end{proof}

\section{Applications}\label{se:appli}

The \emph{treewidth} of graph $G=(V,E)$, denoted $\tw(G)$, is the minimum number $k$ such that $G$ has a triangulation $H=(V,E')$ of clique size at most $k+1$. The \emph{minimum fill in} of $G$ is the minimum size of $F$, over all (minimal) triangulations $H=(V, E \cup F)$ of $G$. The \emph{treelength} of $G$ is the minimum $k$ such that there exists a minimal triangulation $H$, with the property that any two vertices adjacent in $H$ are at distance at most $k$ in graph $G$. 

\begin{proposition}
Let $\Pi_G$ denote the set of \pmc s of graph $G$. The following problems are solvable in $\cO^*(|\Pi_G|)$ time, when $\Pi_G$ is given in the input~: \textsc{(Weighted) Treewidth}~{\cite{FKTV08,BoRo03}}, \textsc{(Weighted) Minimum Fill-In}~{\cite{FKTV08,Gysel13}},
\textsc{Treelength}~\cite{Lokshtanov10}.
\end{proposition}

Let us also recall the \msphit\ problem where, for a fixed integer $t$ and a fixed $\cmsot$ formula $\varphi$, the goal is to find a pair of vertex subsets $X \subseteq F \subseteq V$ such that $\tw(G[F]) \leq t$, $(G[F],X)$ models $\varphi$ and $X$ is of maximum size. 

\begin{proposition}[\cite{FoToVi14}]\label{pr:msphit}
For any fixed integer $t>0$ and any fixed $\cmsot$ formula $\varphi$, problem  \msphit\ is solvable in $\cO(|\Pi_G| \cdot n^{t+4})$ time, when $\Pi_G$ is given in the input.
\end{proposition}

Pipelined with Theorems~\ref{th:pmc_vc} and~\ref{th:pmcmw}, we deduce:
\begin{theorem}\label{th:appli}
Problems  \msphit\ \textsc{(Weighted) Treewidth}, \textsc{(Weighted) Minimum Fill-In},  \textsc{Treelength} can be solved in time $\cO^*(4^{vc})$ and in time $\cO^*(\pmcb^{\mw})$.\\
\end{theorem}
We re-emphasis that problem \msphit\ generalizes many classical problems, e.g.,  \textsc{Maximum Independent Set},  \textsc{Maximum Induced Forest},  \textsc{Longest Induced Path}, \textsc{Maximum Induced Matching}, \textsc{Independent Cycle Packing},  \kip,\ \kit, \textsc{Maximum Induced Subgraph With a Forbidden Planar Minor}. More examples of particular cases are given in Appendix~\ref{ap:appli} (see also~\cite{FoToVi14}). 

The polynomial factors hidden by the $\cO^*$ notation depend on the problem and on the parameter, they are typically between $n^5$ to $n^7$.

\section{Conclusion}

We have provided single exponential upper bounds for the number of minimal separators and the number of potential maximal cliques of graphs, with respect to parameters vertex cover and modular width. 

A natural question is whether these results can be extended to other natural graph parameters. We point out that for parameters like clique-width or maximum leaf spanning tree, one cannot obtain upper bounds of type $\cO^*(f(k))$ for any function $f$. A counterexample is  provided by the graph $W_{p,q}$, formed by $p$ disjoint paths of  $q$ vertices plus two vertices $u$ and $v$ seeing the left, respectively right ends of the paths (similar to the watermelon graph of Figure~\ref{fi:cubewaterm}). Indeed this graph has a maximum leaf spanning tree of $p$ vertices and a clique width of no more than $2p+1$, but it has roughly $p^{n/p}$ minimal $u,v$-separators. 

Finally, we point out that our bounds on the number of \pmc s w.r.t. vertex cover and to modular width do not seem to be tight. Any improvement on these bounds will immediately provide improved algorithms for the problems mentioned in Section~\ref{se:appli}.
\bibliographystyle{siam}
\bibliography{VC_PMC.bib}
 
\appendix

\section{More applications}\label{ap:appli}
We give in this Appendix several problems that are all known to be particular cases of \msphit (see~\cite{FoToVi14} proofs and more applications). Proposition~\ref{pr:msphit} also extends to the weighted version and the annotated version of the problems (in the annotated version, a fixed vertex subset must be part of the solution $F$).

  Let   $\cF_m$  be the set of cycles of length  $0\imod{m}$.
   Let  $\ell \geq 0$ be an  integer.  Our first example is the following  problem.
   
   \medskip
\defproblem{\amims{}}{A graph $G$.}{Find a  set $F\subseteq V(G)$ of maximum size such that  
   $G[F]$ contains at most $\ell$ vertex-disjoint cycles  from ${\cal F}_m$.}
 \medskip 
  
   \amims{} encompasses several interesting problems.
   For example, when $\ell=0$, the problem is to find a maximum induced subgraph without cycles divisible by $m$. For $\ell=0$ and $m=1$ this is \textsc{Maximum Induced Forest}. 
   
   For  integers $\ell \geq 0 $ and $p\geq 3$, the problem related to \amims{} is the   following.
   
   \medskip
\defproblem{\mislkc{}}{A graph $G$.}{Find a  set $F\subseteq V(G)$ of maximum size such that  
   $G[F]$ contains at most $\ell$ vertex-disjoint cycles  of length at least $p$.
   }
   
   \medskip 
   
Next example concerns properties described by forbidden minors. 
    Graph $H$ is a \emph{minor} of graph $G$   if $H$ can be obtained from a
subgraph of $G$ by a (possibly empty) sequence of edge contractions.  A \emph{model} $M$ of minor $H$ in $G$ is a minimal subgraph of $G$, where the edge set $E(M)$ is partitioned into \emph{c-edges (contraction edges)} and \emph{m-edges (minor edges)} such that the graph resulting from contracting all c-edges is isomorphic
 to $H$.
Thus, $H$ is isomorphic to a minor of $G$ if and only if there exists a model
of $H$ in $G$. For  an integer $\ell$ a finite set of graphs ${\cF_{plan}}$, containing a planar graph we define he following generic problem.
 
\medskip


 \defproblem{\mislf{} }{A graph $G$.}{Find a set  $F \subseteq V(G)$
  of maximum size such that  $G[F]$ contains at most $\ell$ vertex disjoint  minor models of  graphs  from
    ${\cF_{plan}}$ }

Even the special case with $\ell=0$, this problem and its complementary version called the \fd, encompass many different problems.
%
%
%

%
  
  \medskip

   Let $t\geq 0$ be an integer and $\varphi$ be a CMSO-formula. 
  Let $\mathcal{G}(t,\varphi)$ be a class of connected graphs of treewidth at most $t$ and with property expressible by $\varphi$.
Our next example is  the following problem.
%


%

\defproblem{\igp{}}{A graph $G$.}{Find a  set $F \subseteq V(G)$ with  maximum number of connected components such that  
   each connected component of $G[F]$ is in $\mathcal{G}(t,\varphi)$.}

As natural sub cases studied in the literature we can cite~\textsc{Independent Triangle Packing} or~\textsc{Independent Cycle Packing}.
\medskip

%

The next problem is an example of \emph{annotated version} of optimization problem \msphit.

\defproblem{\kig{}}{A graph $G$,  with $k$ terminal vertices.}{Find   an induced graph from   $\mathcal{G}(t,\varphi)$ containing all $k$ terminal vertices.}

\medskip
Many variants of \kig{} can be found in the literature, like  \kip,\ \kit, \kic.

\section{Monadic Second-Order Logic}\label{ap:logic}

We use Counting Monadic Second Order Logic ($\cmsot$), an extension of $\msot$, as a basic tool to express properties of vertex/edge sets in graphs. 
 \smallskip

The syntax of Monadic Second Order Logic ($\msot$) of graphs includes the logical connectives $\vee,$ $\land,$ $\neg,$ 
$\Leftrightarrow ,$  $\Rightarrow,$ variables for 
vertices, edges, sets of vertices, and sets of edges, the quantifiers $\forall,$ $\exists$ that can be applied 
to these variables, and the following five binary relations: 
\begin{enumerate}

\item 
$u\in U$ where $u$ is a vertex variable 
and $U$ is a vertex set variable; 
\item 
 $d \in D$ where $d$ is an edge variable and $D$ is an edge 
set variable;
\item 
 $\mathbf{inc}(d,u),$ where $d$ is an edge variable,  $u$ is a vertex variable, and the interpretation 
is that the edge $d$ is incident with the vertex $u$; 
\item 
 $\mathbf{adj}(u,v),$ where  $u$ and $v$ are 
vertex variables  and the interpretation is that $u$ and $v$ are adjacent; 

\item 
 equality of variables representing vertices, edges, sets of vertices, and sets of edges.
\end{enumerate}

The $\msoo$ is a restriction of $\msot$ in which one cannot use edge set variables (in particular the incidence relation becomes unnecessary). For example \textsc{Hamiltonicity} is expressible in $\msot$ but not in $\msoo$. 

 In addition to the usual features of monadic second-order logic, if we have atomic sentences testing whether the cardinality of a set is equal 
to $q$ modulo $r,$ where $q$ and $r$ are integers such that $ 0\leq q<r $ and $r\geq 2,$ then 
this extension of the $\msot$ (resp. $\msoo$) is called the {\em counting monadic second-order logic} $\cmsot$ (resp. $\cmsoo$). So essentially $\cmsot$ (resp. $\cmsoo$)
is $\msot$ (resp. $\msoo$) with the following atomic sentence for a set $S$: 
\begin{quote}
$\mathbf{card}_{q,r}(S) = \mathbf{true}$ if and only if $|S| \equiv q \pmod r.$ 
\end{quote}

We refer to~\cite{ArnborgLS91,Courcelle90} and the book of Courcelle and Engelfriet~\cite{CoEn12} for a detailed introduction on different types of logic.
 \end{document}